\documentclass[conference]{IEEEtran}

\usepackage{amsmath}
\usepackage{graphicx}
\include{psfig}
\usepackage{epsfig}
\usepackage{array}
\usepackage{psfrag}
\usepackage{amssymb}
\usepackage{color}
\usepackage{pstricks,pst-node,pst-text,pst-3d,pst-plot}
\usepackage{cite}

\usepackage{amsmath}
\usepackage{graphics}
\usepackage{trig}
\usepackage{graphicx}
\usepackage{amsmath, amsthm, amssymb}
\usepackage{epsf}
\usepackage{psfrag}
\usepackage{pstricks}
\usepackage{pst-node}
\usepackage{pst-3d}
\usepackage{pst-plot}
\usepackage{soul}
\usepackage{pbox}

\usepackage{array}
\usepackage{amssymb}
\usepackage{color}
\usepackage{algpseudocode}
\usepackage{algorithm}
\usepackage{enumerate}
\usepackage{wrapfig}
\usepackage{epstopdf}

\theoremstyle{remark}

\theoremstyle{lemma}
\newtheorem{lemma}{Lemma}[section]

\def\bx{{\bf x}}

\def\bn{{\bf n}}
\def\bt{{\bf t}}

\def\bW{{\bf W}}

\def\bh{{\bf h}}

\def\bH{{\bf H}}

\def\bQ{{\bf Q}}

\def\bw{{\bf w}}
\def\bR{{\bf R}}
\def\br{{\bf r}}

\def\bA{{\bf A}}
\def\ba{{\bf a}}

\def\bx{{\bf x}}

\def\bk{{\bf k}}

\def\bSigma{{\boldsymbol \Sigma}}
\def\diag{{\rm diag}}

\renewcommand{\baselinestretch}{0.993}
\begin{document}



 \title{ A Coalition Formation Approach to Coordinated Task Allocation in Heterogeneous UAV Networks}

\author{
   \IEEEauthorblockN{
    Fatemeh Afghah\IEEEauthorrefmark{1},
     Mohammad Zaeri-Amirani\IEEEauthorrefmark{1},
     Abolfazl Razi\IEEEauthorrefmark{1}, 
    Jacob Chakareski\IEEEauthorrefmark{2}, 
    and
     Elizabeth Bentley\IEEEauthorrefmark{3}}
   \IEEEauthorblockA{
     \IEEEauthorrefmark{1}School of Informatics, Computing and Cyber Systems,
    Northern Arizona University, Flagstaff, AZ, United States\\
     Email: \{fatemeh.afghah,Mohammad.Zaeri-Amirani,abolfazl.razi\}@nau.edu}
   \IEEEauthorblockA{
    \IEEEauthorrefmark{2}Department of Electrical \& Computer Engineering,
    University of Alabama, Tuscaloosa, AL, United States\\
     Email: jacob@ua.edu}
   \IEEEauthorblockA{
    \IEEEauthorrefmark{3}Air Force Research Laboratory,
    Rome, NY, United States\\
     Email: elizabeth.bentley.3@us.af.mil}
 }
 \maketitle

\begin{abstract}
The problem of adversary target detection and the subsequent task completion using a heterogeneous network of resource-constrained UAVs is considered. No prior knowledge about locations and required resources to identify these targets is available to the UAVs. In the proposed leader-follower coalition formation model, the UAV that first locates a target serves as the coalition leader and selects a group of follower UAVs to complete the task associated with the identified target. The goal of the coalition formation is to complete the designated tasks with minimal resource utilization. Another role of coalition members is to make the ground station aware of the detected adversary target by forwarding its signal to the station via a distributed cooperative relaying scheme. We also propose a reputation-based mechanism for coalition formation to monitor the cooperative behavior of the UAVs over the course of time and exclude potentially untrustworthy UAVs. Simulation results show the efficiency of the proposed method in forming optimal coalitions compared to alternative methods.\footnote{ Distribution Statement A: Approved for Public Release; distribution unlimited: 88ABW-2016-4923 on 04 Oct 2016 }

\end{abstract}

Index Terms-  Coalition formation, task coordination, unknown environment, cooperative communication, UAV networks, merge-and-split.

\IEEEpeerreviewmaketitle
\section{Introduction}
Recent advancements in communication and computation systems allow deployment of large teams of Unmanned Aerial Vehicles (UAVs) to cooperatively accomplish complex missions that often cannot be performed by a single UAV. Several features including task coordination and reliable communications are required to enable interoperability within the heterogeneous airborne networks, particularly for autonomous operations and providing on-board data processing during the mission \cite{Razi_WiSEE,Razi_Asilomar}. These heterogeneous autonomous vehicle systems could provide a great flexibility to complete compound tasks which are distributed in time and space. Target detection, data collection, target tracking and prosecution, imaging, and surveillance are typical examples of tasks that can be accomplished by a heterogeneous UAV network.

The task allocation problem in multi-agent systems is defined as the process of allocating a set of tasks to groups of agents to ensure timely and efficient task completions, noting the individual capabilities of the agents \cite{Korsah}. Several studies have addressed the complex problem of task assignment using different approaches including mixed integer linear programming \cite{Schumacher,Darrah}, dynamic network flow optimization \cite{Schumacher2,Schumacher3}, market-based strategy \cite{Dias,Gerkey}, finite state machine \cite{Zhong}, and multiple choice knapsack problem \cite{Alighanbaril}.
While in the majority of the aforementioned works, the tasks are centrally assigned to the agents by a base station that has complete knowledge about the tasks and often the agents' capabilities; in many dynamic systems, tasks may appear at unpredictable locations and times (e.g. target detection in army fields, search and rescue operations). Hence, a priori knowledge about these tasks is not always available to the base station. Even when such centralized task allocation algorithms exist, they are often computationally intensive even in homogeneous networks and not easily scalable to systems with a large number of tasks or agents \cite{Zhang}. Noting the computation and communication capabilities of modern devices, the agents can be considered as smart entities with decision-making capabilities. Such cognitive capability facilitates implementation of distributed task allocation mechanisms by allowing the agents to observe the environment and monitor the operation of other agents and properly respond to the observed situations.

Coalition formation game is a class of games, in which the players cooperate with each other by forming various sub-groups called coalitions. This class of games has been recently used in various applications such as task assignment in multi-agent systems, and communication networks \cite{CISS_Ashwija,Bayram,Arslan,Vig,Saad_UAV}.
In this paper, we study the problem of cooperative task completion in a network of heterogeneous UAVs with constrained individual resources. We assume that a number of targets are distributed in an unknown environment, where no prior information about the targets' time of appearance, and location is available. We also assume that there is one compound task associated with each target. The tasks can differ intrinsically based on the characteristics of their encountered targets in terms of mobility, speed, position, and the required resources. For instance, prosecution of a fighter tank or a long-range missile in a battle field require different sets of equipment.
The objectives of our proposed model include: i) locating the distributed targets, ii) identifying their associated tasks and required resources, and iii) completing the identified tasks. In order to accomplish these goals, we form several coalitions of UAVs using a coalition formation method, in which each coalition will complete a task associated to one target. The proposed coalition formation model with a leader-follower structure is designed so that it ensures providing adequate resources to complete each encountered task while minimally exceeding the minimum required resources. The proposed dynamic coalition formation model enables the UAVs to overcome the limitations of their individual capabilities such as limited payload and computation, and communication resources. Formation of UAV coalitions can also extend the coverage area in target tracking and surveillance applications compared to utilizing individual UAVs.  
The traveling distances of the UAVs to the tasks are also taken into account in forming optimal coalitions to complete the tasks in a timely manner. Furthermore, in order to get more detailed information about the identified targets (e.g. an adversary object in a battlefield), the UAVs in each coalition are required to forward the broadcast messages by their encountered target to the base station using a beamforming technique.

In majority of previously reported works on coalition formation for task allocation, the objective is to enhance the efficiency of the formed groups in task performance noting the different capabilities available at the coalition members \cite{Shehory,Service2011,Dang,Vig,Bayram,Chen}. Therefore, the agents often consider a solution to be optimal when it maximizes the total utilities of the group in executing the existing tasks with minimum resources. In these works, it is assumed that all the agents are fully trustable, and they are obligated to cooperate with one another by utilizing their initially claimed resources to complete the tasks. However, this assumption is far too optimistic since cooperation is not an inherent characteristic of cognitive but potentially self-interested agents \cite{Afghah_NWRCS,Afghah_CDC}. In this work, we consider a commercial scenario, in which the UAVs can belong to different vendors, and they can collect some sort of monetary benefits for participation in each mission. In the proposed leader-follower coalition formation model, we account for both the efficiency in task performance considering the resource constraints during the coalition formation from the leaders' perspective, as well as the individual preferences of the followers during the followers' decision making to join the available coalitions. In such realistic commercial setting, the UAVs may exhibit selfish behaviors by not utilizing the resources that they originally committed during the coalition formation, with the incentive of saving these resources for future missions to obtain higher benefits. In our proposed model, we develop a novel reputation-based mechanism that keeps the record of the UAVs' cooperative behaviors. The cumulative credits for UAVs are used to identify the trustable UAVs during the member selection procedure in coalition formation. This definition involves a trade-off scenario for the UAVs. On one hand, they prefer to avoid resource sharing with others to save their limited available resources, while on the other hand they need to cooperate with other agents and sustain a good reputation in order to be selected for next missions.

The rest of this paper is organized as follows: In Section \ref{sec:SystemModel}, an overview of the system model for heterogeneous UAV network is provided and the proposed coalition formation game-theoretic model is described. The optimization problem to determine the optimal beamforming vector is described in Section \ref{sec:com}. Simulation results are provided in Section \ref{sec:simulation}, followed by concluding remarks in Section \ref{sec:Conclusion}.

\textbf{Notation:} Vectors and matrices are represented by lower case and upper case bold letters, respectively. The notations $(.)^*$, $(.)^T$ and $(.)^H$ demonstrate the conjugate, transpose and conjugate transpose (Hermitian) operations, respectively. The real value and imaginary values are shown by $\mathcal{R}\{.\}$ and $\mathcal{I}\{.\}$. The diagonal matrix $\bA = \diag\{\ba\}$ is a matrix whose diagonal elements are the elements of the vector $\ba$. Finally, the function $\gamma_{L,\epsilon}(x)$ for $x>0$ is defined as
$\left\{\begin{array}{cc}
         x & x \leq 1 + \epsilon\\
         L & \mbox{Otherwise}
       \end{array}\right.$, for a given value of $L$ and small value of $\epsilon$, and we have $\gamma_{L}(x)=\gamma_{L,0}(x)$.

\begin{figure}[t]
  \centering
  \centerline{\includegraphics[width=8. cm]{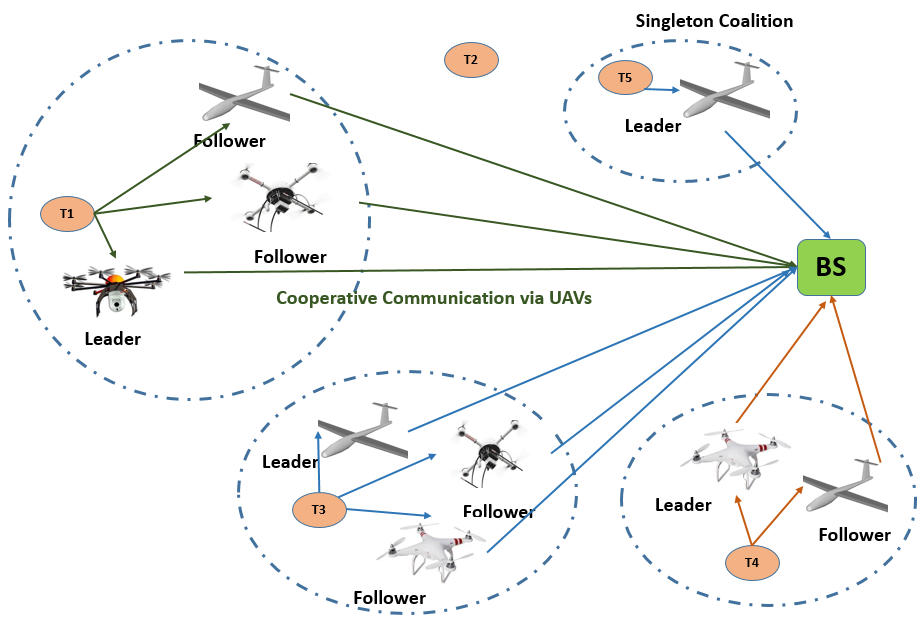}}
 \vspace{-5 pt}
\caption{Leader-follower coalition formation in a heterogeneous UAV network to prosecute unknown targets}
\label{fig:systemmodel}
 \vspace{-15 pt}
\end{figure}

\section{Coalition Formation for Joint Task Allocation and Communication Optimization} \label{sec:SystemModel}
A heterogeneous system of $N$ UAVs, $\mathcal{U} = \{U_1,U_2,\dots,U_{N}\}$ is considered, where the UAVs can form various coalitions to accomplish the dynamic tasks in the network, as depicted in Fig. \ref{fig:systemmodel}. Suppose that the $i^{\text{th}}$ UAV is assigned with a capability vector $\br_i = [r_i^1,r_i^2, \ldots, r_i^{N_r}]^T$ with $r_i^j \geq 0$, $j = 1, 2, \ldots, N_r$, that identifies the available resources at this agent assuming that there are $N_r$ different types of resources in the network. Each resource can be either consumable or non-consumable. If the $i^{\text{th}}$ UAV does not have resource $j$, then $r_i^j = 0$, and if resource $j$ is non-consumable, then $r_i^j = \infty$. This vector can vary over time based on the amount of resources the UAVs consume to complete the tasks.

A coalition of users, named $\mathcal{S}_k$, is a non-empty subset of UAVs, $\mathcal{S}_k \subset \mathcal{U}$ that handles task $k$. It is assumed that each coalition of UAVs will handle one target at a time. Moreover, we consider a sparse distribution of targets in the environment, hence noting the potential distances among the formed coalitions, it is assumed that each UAV can be only a member of one of these coalitions at a certain time. Hence, the coalitions are non-overlapping, meaning that $\mathcal{S}_k \cap \mathcal{S}_l=\phi$. The coalition of all UAVs, $\mathcal{U}$ is the grand coalition and a coalition that only contains one UAV is called a singleton coalition. Each coalition $\mathcal{S}_k$ is associated with a vector of available resources which is shown by $\bR_k = [R_k^1,R_k^2,\ldots,R_k^{N_r}]^T$ with $R_k^{j} \geq 0$, $j = 1,2,\ldots,N_r$. For simplicity, one may assume that the vector $\bR_k$ is additive over the elements' resource vectors in the coalition, i.~e.~ $\bR_k = \sum_{j \in \mathcal{S}_k}{\br_i}$.

It is assumed that the UAVs operate in an unpredictable environment, where targets of various types with different resource requirements appear in random time and locations and move freely afterwards. Therefore, the base station is not aware of the potential targets. We assume that all the UAVs have the capability of searching for new targets in a limited geographical field, and the target detection is performed by the UAVs independent of the base station. The procedure of prosecuting a target is defined as a compound task to be accomplished by a coalition (the term task can refer to a set of multiple sub-tasks required to prosecute a target). The tasks can differ inherently based on the characteristics of their corresponding targets; hence the number, duration and location of tasks vary over time. Since each task is associated to a target, we use the terms target and task interchangeably throughout this paper.

When a UAV detects a task $k$, it determines the type and amount of the resources required to carry out this task, i.~e.~ $\Gamma_k = [\tau_{k}^1,\tau_{k}^2,\ldots,\tau_{k}^{N_r}]^T$, where $\tau_{k}^i \geq 0$. If this UAV does not have sufficient resources to perform the detected task, it calls for a coalition formation and serves as the coalition leader.
Since the ground station has no information regarding the existing targets in the network, the member UAVs of each coalition are required to listen to the radio communication of a target and forward it to the ground station \cite{UAVMilitary2016}. In other words, a coalition of UAVs is supposed to relay the target's message $s_k$ with $E\{s_k^*s_k\} = 1$ to the base station.

In our proposed model, first each leader forms an initial coalition of the available UAVs to maximize the efficiency in task performance for the encountered target, and send a request to join the coalition to the selected followers. Then, the follower UAVs observe the formed coalitions by different leaders and decide to join the coalition that benefits them the most. The details of the proposed coalition formation algorithm is described in Section \ref{sec:formation} as well as in Fig. \ref{fig:statediagram}. In order for a leaders to form a coalition, it takes into account several factors including i) collecting the required resources to perform a task, ii) traveling time of the UAVs to the task, and iii) the quality of service (QoS) of target message's communication at the base station. The goal of this combinational complex optimization problem is to perform the tasks in timely and resource efficient manner. This goal suggests that the members of a newly formed coalitions should collectively have all the required resources to perform the encountered task, while minimally surpassing the task requirements. The latter objective can be considered as the cost of coalition, hence this coalition formation game is  not super-additive. Other associated costs to coalition formation include a higher chance of UAVs' collisions when having more UAVs in a coalition, as well as heavier signaling loads for the members to exchange the necessary information.

Another criterion in coalition formation from the leaders' perspective is the deadline to complete a task. Hence, the leaders need to take into account the traveling time of the UAVs to the task in order to choose the coalition members. For this purpose, $\delta_{i,k}$ is defined as the traveling time of the $i^{\text{th}}$ member of coalition to target $k$ location that should be less than or equal to a preferred threshold $f_{\delta}(\rho)$ as a function of field radius $\rho$.

A key contribution of the proposed model is to identify the reliable UAVs by the leaders to join the coalitions and filter out any UAVs with selfish behavior. \footnote{It is worth mentioning that such networks are prone to suffer from both malicious/intruding UAVs as well as selfish ones. The focus of this work is to prevent potential selfish behavior of legitimate non-altruistic UAVs in not spending their resources after joining a coalition to encourage cooperation among them, while studying the required authentication methods to identify intruding UAVs, or preventing all potential malicious acts of these UAVs (e.g. reporting false information about their available resources or location) is out of the scope of this paper. } This is facilitated by defining a cumulative cooperation credit for each UAV that indicates its cooperative behavior in terms of resource sharing during task completion over the course of time. On one hand, since the leaders prefer to select the potential followers with higher credits, the UAVs are motivated to maintain a good cooperative credit. On the other hand, the follower UAVs similar to other types of cognitive agents may act selfishly and avoid consuming their resources after joining a coalition. The incentive behind this behavior is to save their limited available resources for future missions to earn more  monetary benefits. The details of cooperation credit is defined in section \ref{sec:credit}.

\begin{figure}[b]
  \centering
  \centerline{\includegraphics[width=8 cm]{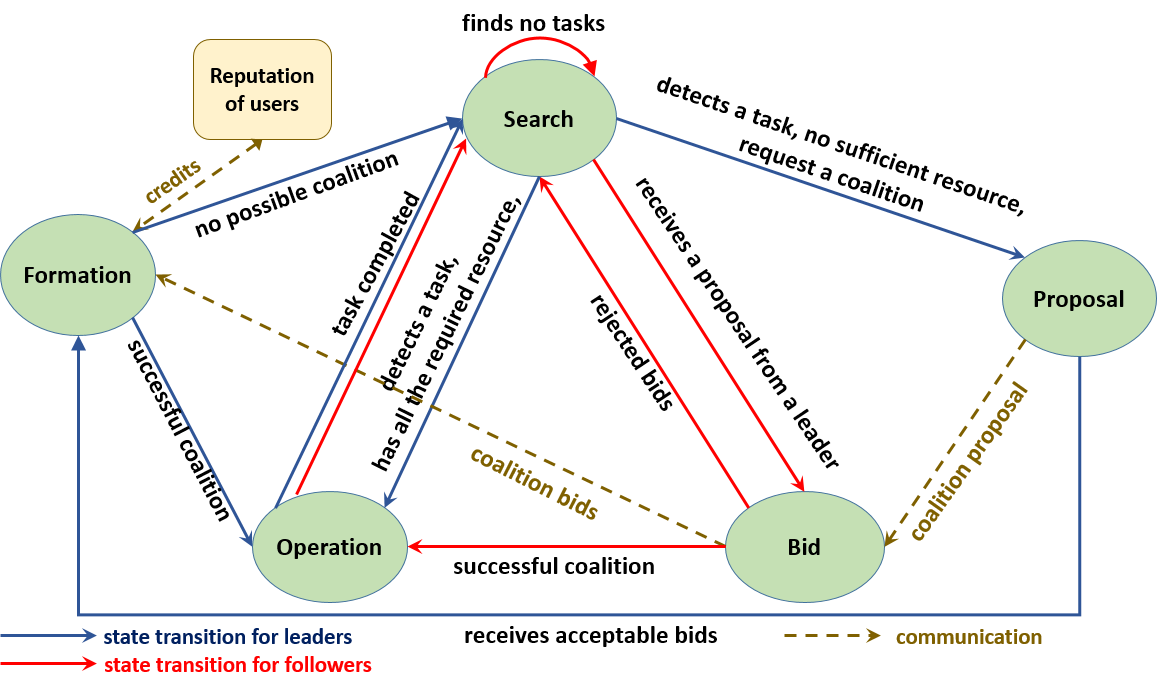}}
\caption{State diagram of the proposed coalition formation.}
\label{fig:statediagram}
\end{figure}
Here we describe the details of the coalition formation process. First, the coalition leader broadcasts a \emph{proposal} to form a coalition. Then, the UAVs who possess at least one of the required resources can respond to this request by reporting their available resources as well as their current positions. This is called the \emph{bid} process. The coalition leader then evaluates all the bids by assessing the resources offered by the volunteers, their estimated arrival time, their cooperative credits as well as the provided QoSs for the relaying services during the \emph{formation} process to determine if a coalition can be formed to complete the encountered task. If a coalition can not be formed, it informs the UAVs, otherwise it gets back to the selected UAVs with information about the tasks (e.g required resources, and location). The UAVs that are not selected by the potential leader go back to the search mode. When a UAV receives multiple coalition formation requests, it considers two factors of expected increase in its cooperative credit based on information it received from the leaders about required resources for the tasks as well as its distance to the targets in order to decide which coalition to join. The coalition formation process is summarized in Fig. \ref{fig:statediagram}. In Section \ref{sec:beam}, the cooperative communication method to relay the target's message to the ground station is described, followed by the definition of the cooperative credit and formulation of the proposed coalition formation algorithm in Sections \ref{sec:credit} and \ref{sec:formation}.


\subsection{Cooperative Communications at Member UAVs of Each Coalition} \label{sec:beam}
Amplify-and-Forward (AF) is a widely used relaying method in communication networks due to its simplicity and low-complexity.  Beamforming is a technique based on AF relaying, where a set of relay nodes amplify and shift the phase of a transmitter's signal and rebroadcast it such that they add up constructively, while interfering signals add up destructively.
In order to relay the target's message to the ground station, a beamforming scenario is proposed in which each UAV in the coalition multiplies the received target's signal to a complex weight number and rebroadcasts it (AF relaying). A common assumption of knowledge of Channel State Information (CSI) in the system is followed in this model\cite{Zaeri2012,zaeri2016}. Assume that $\bh_{TU} = [h_{TU_0},h_{TU_1}, \ldots,h_{TU_{N_{\mathcal{S}}}}]^T$ is the vector of instantaneous channel coefficients between the target and the coalition members, where $N_{\mathcal{S}} = |\mathcal{S}|-1$. Likewise,  $\bh_{UB} = [h_{U_0B},h_{U_1B}, \ldots,h_{U_{N_{\mathcal{S}}}B}]^T$ is the vector of channel coefficients between the coalition members and the base station. Here, the index $0$ represents the leader and the indices $i = 1,2,\ldots,N_{\mathcal{S}}$ represent the coalition members. It is assumed that the leader has the knowledge of instantaneous reciprocal channel vectors, (i.~ e.~ $\bh_{TU}$ and $\bh_{UB}$) and is responsible for calculating the optimum beamforming and notifying the members of its coalition.

If $s_T$ denotes the transmit signal, the vector of the received signal at the coalition can be written as:
\begin{equation}\label{Eq:RecAtUAVs}
  \bx = \bh_{TU} s_T + \bn,
\end{equation}
where $\bn \sim \mathcal{N}(\textbf{0},\bSigma^{\frac{1}{2}})$, with diagonal covariance matrix $\bSigma$, is an additive zero mean Gaussian noise vector at the coalition members. It is assumed that each UAV is aware of its local noise characteristics and performs the whitening process before the beamforming. Hence, the whitened received signal at UAVs can be written as:
\begin{equation}\label{Eq:WhiteRecAtUAVs}
  \widetilde{\bx} = \bSigma^{-\frac{1}{2}}\bx = \bSigma^{-\frac{1}{2}}\bh_{TU} s_T + \bSigma^{-\frac{1}{2}}\bn,
\end{equation}
The $i^{\text{th}}$ coalition member, $i = 0,1, \ldots, N_{\mathcal{S}}$, multiplies the received target signal by a complex weight $w_i^*$ and then relays it. The broadcasted signal by the coalition members can be written as follows, assuming $\bw = [w_0, w_1, \ldots, w_{N_{\mathcal{S}}}]^T$:
\begin{equation}\label{Eq:transmitByUAV}
  \bt = \bW^H\widetilde{\bx} = \bW^H\bSigma^{-\frac{1}{2}}\bh_{TU}s_T + \bW^H\bSigma^{-\frac{1}{2}}\bn,
\end{equation}
where $\bW = \diag{(\bw)}$.
The received signal at the base station is:
\begin{IEEEeqnarray}{rCl}\label{Eq:ReceivedBase}
  y_B &=& \bh_{UB}^H \bW^H \bSigma^{-\frac{1}{2}}\bh_{TU} s_T + \bh_{UB}^H \bW^H \bSigma^{-\frac{1}{2}}\bn + \nu \IEEEnonumber \\
  &=& \bw^H \bH_{UB}^H\bSigma^{-\frac{1}{2}}\bh_{TU} s_T +  \bw^H \bH_{UB}^H \bSigma^{-\frac{1}{2}}\bn + \nu,
    \IEEEeqnarraynumspace \IEEEyesnumber
\end{IEEEeqnarray}
where, $\bH_{UB} = \diag{(\bh_{UB})}$ and $\nu \sim \mathcal{N}(0,\sigma^2)$ is white Gaussian noise at the base station receiver.
Therefore, the Signal to Noise Ratio (SNR) at the base station can be expressed by:
\begin{equation}\label{Eq:SNRBase}
  {SNR}_B = \frac{\bw^H \bk \bk^H\bw}{\bw^H \bH_{UB}^H\bH_{UB}\bw + \sigma^2},
\end{equation}
where $\bk = \bH_{UB}^H\bSigma^{-\frac{1}{2}}\bh_{TU}$. Each UAV (including the leader) has a limited energy to forward the target's signal. Using (\ref{Eq:transmitByUAV}), the power consumption at each coalition member for relaying the target's signal can be written as:
\begin{IEEEeqnarray}{rCl}\label{Eq:individualPower}
&P_{i,R} = [\bw^H]_i[\bSigma^{-\frac{1}{2}}]_{i,i}[\bh_{TU}]_i[\bh_{TU}^H]_i[\bSigma^{-\frac{1}{2}}]_{i,i}[\bw]_i \IEEEnonumber \\
&+ [\bw^H]_i[\bw]_i \IEEEyesnumber
\end{IEEEeqnarray}
where $i = 0,1,2,\ldots, N_{\mathcal{S}}$. We assume that the individual transmission power of UAV $i$ is below a certain threshold, denoted by $P_{i}^{max}$.

\subsection{Cooperation Credit for UAVs} \label{sec:credit}
To monitor the cooperative behavior of the UAVs, a cumulative cooperative credit is defined for each UAV based on the amount of resources that it utilizes for a specific task. We assume that the credit of all UAVs are initialized to an equal initial credit $C_i^{(0)} = C > 0$, $i = 1,2,\ldots,N$. After completing a task, each UAV's credit is calculated using the following steps:

\begin{enumerate}
  \item First, the credit of UAV $i$ at time $n$ is updated as
  \begin{equation}
      \tilde{C}_i^{(n)} =
          \left\{\begin{array}{l c}
               C_i^{(n-1)} + \Delta{C_i}^{(n)} & \mbox{If \;\;} \exists k | ~~U_i \in \mathcal{S}_k, \\
               C_i^{(n-1)} & \mbox{Otherwise}
          \end{array}\right.
  \end{equation}
Here, the change in credit $\Delta{C_i}^{(n)}$ is defined by:
  \begin{equation}
    \Delta{C_i}^{(n)} = \frac{\tau_k}{\sum_{l \in \mathcal{S}_k} a_l} a_i
  \end{equation}
where $\tau_k = \sum_{j = 1}^{N_r}\tau_k^j$ represents a value of task k and $a_i = \sum_{j = 1}^{N_r} \gamma_1(\frac{r_i^j}{\tau_k^j})$ denotes the effective resource contribution of each UAV.

\item If all $\tilde{C}_i^{(n)}$ for $i = 1,2,\ldots,N$ are equal, then we set $C_i^{(n)} = C$ to reset the credits. Otherwise:
  \begin{equation}
    {C}_i^{(n)} = C \frac{\tilde{C}_i^{(n)} - \tilde{C}_{min}}{\tilde{C}_{max}-\tilde{C}_{min}}
  \end{equation}
where $\tilde{C}_{max} = \max_{i=1,2,\ldots,N}{\{\tilde{C}_i^{(n)}\}}$ and $\tilde{C}_{min} = \min_{i=1,2,\ldots,N}{\{\tilde{C}_i^{(n)}\}}$.
\end{enumerate}
Therefore, the credits in each step are scaled to values within $[0,C]$ range. After the completion of each task, the updated credits are broadcasted by the coalition leader to be used for future coalition formations. For the sake of simplicity, we drop the superscripts $(n)$ hereafter.

\subsection{Proposed Leader-Follower Coalition Formation} \label{sec:formation} 
In the proposed coalition formation method, first the leaders select their coalition members among the available candidates to maximize their corresponding coalition's utility function to enhance the efficiency in completing their encountered task noting the resource constraints. The coalition value for the $k^{\text{th}}$ leader is defined in such a way to: i) assure the existence of required resources to handle a task, ii) avoid over-spending the resources on a specific task, iii) guarantee the timely completion of the task, iv) provide the required quality of communication to relay target's message, and also v) select the reliable UAVs as follow:
\begin{IEEEeqnarray}{rCl}\label{Eq:LeaderUtility}
  &v(\mathcal{S}_k) = \alpha_1 \sum_{i = 1}^{N_{\mathcal{S}}} C_i+\alpha_2 \gamma_{-L}\Big(\frac{SNR_{Thr}}{SNR_{\mathcal{S}_k}}\Big)\nonumber \\
   &+ \alpha_3 \sum_{j = 1}^{N_r}{\gamma_{-L} \Big(\frac{R_k^j}{\tau_k^j} \Big)}
  - \gamma_{L} \Big(\frac{\max_{i|U_i \in \mathcal{S}_k}\{\delta_{i,k}\}}{f_{\delta}(\rho)}\Big)
  \IEEEeqnarraynumspace
\end{IEEEeqnarray}
for a large enough positive value of $L$. Here, the design parameters $\alpha_1, \alpha_2$, and $\alpha_3> 0$ represent the importance of the credit and quality of communication compared to the resource optimization goal for the leader. Also, $SNR_{Thr}$ represents the minimum required $SNR_{\mathcal{S}_k}$ to successfully relay the target's message to the base station. It is worth mentioning that the maximum function in last term of (\ref{Eq:LeaderUtility}) is used to ensure that even the UAV with the latest time of arrival will be at task position in time. It should be noted that the k'th leader is a member of the coalition $\mathcal{S}_k$. 

\begin{algorithm}
\caption{Coalition formation algorithm} \label{Table:formation}
\begin{algorithmic}[1] 
\State{Initializing coalitions} \Comment{Each leader $k$ starts from a partition of the singleton coalitions of UAVs who responded to its proposal}
\State \emph{Merge-and-split coalition formation algorithm by the leader}
\Comment{for all available UAVs for this task}
\While{Change in coalition values is greater than $\epsilon$}
\While{$\mathcal{S}_i$ and $\mathcal{S}_j$ exist with
        $v(\mathcal{S}_i\bigcup \mathcal{S}_j) > v(\mathcal{S}_i) + v(\mathcal{S}_j)$}
\State Merge $\mathcal{S}_i$ and $\mathcal{S}_j$.
\EndWhile
   \While{$\mathcal{S}_i$ and $\mathcal{S}_j$ exist such that:
        $v(\mathcal{S} = \mathcal{S}_i \bigcup \mathcal{S}_j) < v(\mathcal{S}_i) + v(\mathcal{S}_j)$}
    \State  Split $\mathcal{S}$ into partitions $\mathcal{S}_i$ and $\mathcal{S}_j$.
  \EndWhile
 \If{ There is a split}
   \State   go to 4.
    \EndIf
\EndWhile
\State Select the coalition with highest coalition value, notify the UAVs of this selected coalition
\If{All selected potential followers said Yes}
\State Terminate
\Else
\State Exclude the ones with \emph{No} response, go to 3
\EndIf
\newline
\State \emph{Selecting the best formation request from different leaders by each selected follower}
\If{Received only one formation request}
\State Say \emph{Yes} to that leader
\Else
\State {Say \emph{Yes} to the leader of coalition that maximizes (\ref{follower_utility}), say \emph{No} to other leaders}
\EndIf
\end{algorithmic}
\end{algorithm}

In order to find the optimal coalition which maximizes the leader's utility function (\ref{Eq:LeaderUtility}), a search over all $2^{L_k}$ possible coalitions is required, where $L_k$ denotes the number of potential follower UAVs who responded to the proposal of leader $k$. To avoid such extensive search, a low complexity \textit{merge-and-split} algorithm is proposed. In this method, each leader $k$ separately starts from an initial state where the set of UAVs who responded to its proposal is partitioned into $L_k$ singleton coalitions.
Afterward, in each step, two chosen coalitions $\mathcal{S}_{ki}$ and $\mathcal{S}_{kj}$ are merged if $v(\mathcal{S}_{ki} \bigcup \mathcal{S}_{kj}) > v(\mathcal{S}_{ki}) + v(\mathcal{S}_{kj})$. Here, since the value of the coalitions which does not include the leader is zero, the only possible merge happens if a singleton coalition $\{u\}$ and the coalition $\mathcal{S}_k$ which contains k'th leader satisfy the condition $v(\mathcal{S}_k \bigcup \{u\}) > v(\mathcal{S}_k) $. Also, if for a non-singleton coalition $\mathcal{S}$ there exists a partition of two coalitions $\mathcal{S}_{ki}$ and $\mathcal{S}_{kj}$ such that $v(\mathcal{S}= \mathcal{S}_{ki} \bigcup \mathcal{S}_{kj}) < v(\mathcal{S}_{ki}) + v(\mathcal{S}_{kj})$, then $\mathcal{S}$ splits into $\mathcal{S}_{ki}$ and $\mathcal{S}_{kj}$ \footnote{Each leader considers itself as a constant member of all coalitions under evaluation.}. At each step of the merge and split algorithm, the optimum value of the $SNR_B$ in the utility function (\ref{Eq:LeaderUtility}) for the coalitions should be calculated, as described in details in the next section. It is worth mentioning that the merge and split algorithm is used to obtain a suboptimal coalition solution of (\ref{Eq:LeaderUtility}) with lower complexity. When the changes in coalition values over consequent rounds become below a threshold, the coalition with highest coalition value would be selected from the leader's perspective and the members will be notified.

When the optimal coalitions from the leaders' perspective are formed, the formation requests will be sent out to the selected UAVs. If a potential follower receives multiple requests from different leaders, it prefers to join the coalition which benefits it the most. The utility function of the followers is defined as:
\begin{equation} \label{follower_utility}
  v_p(U_i,\mathcal{S}_k) = \Delta{C_i} - \alpha_4 \delta_{i,k},
\end{equation}
where $\Delta{C_i}$ and $\delta_{i,k}$ are the expected credit (knowing the required resources for the encountered task) and traveling time to task $k$ if user $U_i$ joins the coalition $\mathcal{S}_k$, respectively. $\alpha_4$ is a design parameter that indicates the importance of the traveling cost compared to the change in credit. The proposed coalition formation algorithm is summarized in Algorithm \ref{Table:formation}.

\textbf{Stability of the proposed coalition formation algorithm:} The proposed coalition formation process includes a series of merge-and-split coalition formation steps. After each stage of merge-and-split coalition formation, the leader sends out the requests and collects the followers' responses. If all responses are affirmative, the algorithm stops; otherwise, a new merge-and-split is executed that keeps the current members with positive responses, and evaluates the new available UAVs. As such, in order to show the stability of the algorithm, it is sufficient to show that: i) the number of sequential rounds of the algorithm is finite, and ii) each merge-and-split stage is stable. The first is ensured, because at each round, we exclude the members with no interest in joining the formed coalition, and therefore after at most $m_s$ iterations the algorithm stops, where $m_s$ is the number of UAVs with negative response to join a coalition. The second condition is also satisfied since as proved in \cite{Apt}, the formed coalitions by the leaders are $\mathcal{D}_{hp}$-stable. This is due to the fact that the only type of allowed membership changes are based on single or possibly multiple merge-and-splits (i. e. a UAV or a group of them are only allowed to leave a partition by means of merges or splitting).

In Section \ref{sec:com}, the details of the inner optimization problem to determine the optimal SNR of target's signal at the base station for a coalition of interest is described. 

\section{Communication Optimization} \label{sec:com}
The optimization from the leaders' perspective is to maximize the coalition value (\ref{Eq:LeaderUtility}) by searching over the coalitions, i.~e.~ $\mathcal{S}$, and determine the optimum beamforming scheme, i.~e.~ $\bw$ to optimize $SNR_B$.
For each coalition, the optimal value of (\ref{Eq:SNRBase}), ${SNR}_{\mathcal{S}}^{opt}$ can be obtained via the following inner optimization problem:
\begin{IEEEeqnarray}{rCl}\label{Eq:InnerSNRopt}
   & \max_{\bw}\;\;\frac{\bw^H \bk \bk^H \bw}{\bw^H \bH_{UB}^H\bH_{UB}\bw + \sigma^2}   \\
   &[\bw^H]_i\left([\bSigma^{-\frac{1}{2}}]_{i,i}[\bh_{TU}]_i[\bh_{TU}^H]_i[\bSigma^{-\frac{1}{2}}]_{i,i}+1\right)[\bw]_i \nonumber \\ & \leq P_{i}^{max} \;\;\;, i = 0,1, \ldots, N_c \nonumber
\end{IEEEeqnarray}
A bisection method is described in \cite{palomar2010convex} to solve the fractional Quadratically Constrained Quadratic Program (QCQP)s such as the optimization problem in (\ref{Eq:InnerSNRopt}) that can be rewritten as:
\begin{IEEEeqnarray}{rCl}\label{Eq:InnerSNRoptBisection1}
   &\max_{\bw,t} \; t \IEEEyesnumber\\
   & [\bw^H]_i\left([\bSigma^{-\frac{1}{2}}]_{i,i}[\bh_{TU}]_i[\bh_{TU}^H]_i[\bSigma^{-\frac{1}{2}}]_{i,i}+1\right)[\bw]_i \nonumber \\ &\leq P_i^{max} \;\;\;, i = 0,1, \ldots, N_c \IEEEnonumber \\
  & \bw^H \bk \bk^H \bw \geq t\bw^H \bH_{UB}^H\bH_{UB}\bw + t\sigma^2 \IEEEnonumber ,
\end{IEEEeqnarray}
where t is an auxiliary variable. In the bisection method, given the value of $t \geq 0$, the following feasibility check problem is investigated:
\begin{IEEEeqnarray}{rCl}\label{Eq:InnerFeasibility}
   & \mbox{Find}\;\;\bw \IEEEyesnumber\\
 &[\bw^H]_i\left([\bSigma^{-\frac{1}{2}}]_{i,i}[\bh_{TU}]_i[\bh_{TU}^H]_i[\bSigma^{-\frac{1}{2}}]_{i,i}+1\right)[\bw]_i \IEEEnonumber \\ &\leq P_i^{max} \;\;\;, i = 0,1, \ldots, N_c \IEEEnonumber \\
  & \bw^H \bk \bk^H \bw \geq t\bw^H \bH_{UB}^H\bH_{UB}\bw + t\sigma^2 \IEEEnonumber .
\end{IEEEeqnarray}
\normalsize
If the problem described in (\ref{Eq:InnerFeasibility}) is feasible, then the optimal solution of problem (\ref{Eq:InnerSNRoptBisection1}), i.~ e.~ ${SNR}_{\mathcal{S}}^{opt}$, is less than or equal $t$; otherwise ${SNR}_{\mathcal{S}}^{opt} \geq t$. The variable $t$ is up limited by $t^{up}$ which is obtained in the following lemma.
\begin{lemma}\label{Lemma:OneEigen}
$t^{up} = \bk^H \bQ^{-1}\bk$ is an up limit for objective in optimization problem (\ref{Eq:InnerSNRopt}).
\end{lemma}
\begin{proof} Please see Appendix \ref{Appendix:One}.
\end{proof}
Using Lemma 1 and by assuming $\delta$ as the absolute precision of the final objective function value, the iteration's complexity order of the bisection method can be written as $\mathcal{O}(\log(\frac{t^{up}}{\delta}))$. The feasibility problem (\ref{Eq:InnerFeasibility}) can be solved using the SemiDefinite Programming (SDP) relaxation method. A drawback of this method is dealing with matrix rank deduction in SDP relaxation problems which increases the complexity of the solution. As an alternative, \cite{zaeri2016} provides a technique to convert such problems to the Second Order Cone Programming (SOCP) problems. The key point is to notice if $\bw^{opt}$ is an optimal solution to the feasibility problem (\ref{Eq:InnerFeasibility}), then for any arbitrary real number $\theta$, $\tilde{\bw} = \bw^{opt} e^{j\angle{\theta}}$ is also an optimal solution. Therefore, it is possible to assume that $\bw^H \bH_{UB}^H\bSigma^{-\frac{1}{2}}\bh_{TU}$ is a non-negative real number. By considering this constraint, the feasibility problem (\ref{Eq:InnerFeasibility}) can be rewritten as the following SOCP problem for a given $t \geq 0$:
\begin{IEEEeqnarray}{rCl}\label{Eq:InnerFeasibility1}
   & \mbox{Find\;\;}\bw \IEEEyesnumber\\
  & [\bw^H]_i\left([\bSigma^{-\frac{1}{2}}]_{i,i}[\bh_{TU}]_i[\bh_{TU}^H]_i[\bSigma^{-\frac{1}{2}}]_{i,i}+1\right)[\bw]_i \IEEEnonumber \\ &\leq P_i^{max} \;\;\; i = 0,1, \ldots, N_c \IEEEnonumber \\
  & \bw^H \bk \geq \sqrt{t\bw^H \bH_{UB}^H\bH_{UB}\bw + t\sigma^2} \nonumber \\
  & \mathcal{R}\{\bw^H \bH_{UB}^H\bSigma^{-\frac{1}{2}}\bh_{TU}\} \geq 0 \IEEEnonumber \\
  & \mathcal{I}\{\bw^H \bH_{UB}^H\bSigma^{-\frac{1}{2}}\bh_{TU}\} = 0 \IEEEnonumber.
\end{IEEEeqnarray}
The feasibility check problem (\ref{Eq:InnerFeasibility1}) can be effectively solved using the \textit{cvx} convex optimization toolbox. By considering the cubic complexity order of the SOCP method, the complexity order of the communication optimization can be expressed as $\mathcal(|\mathcal{S}|^3 \log(\frac{t^{up}}{\delta}))$.

\section{Simulation results} \label{sec:simulation}
In order to validate the performance of proposed method, a scenario consisting of two leader- and six follower- UAVs is considered. The UAVs are uniformly located in a $\mathcal{R}^3$ region. It is assumed that the leaders do not carry out the required resources to perform the identified tasks individually and they call out to form coalitions. Five types of resources are considered and the amount of each resource at each UAV as well as the required resources for each task are generated randomly. The traveling time is generated proportional to the distance of the UAVs to the corresponding target. All communication channels are generated by zero mean Gaussian variables with variance proportional to the inverse distance between the corresponding transmitter and receiver antennas.

\begin{figure}[t]
\vspace{-10pt}
  \centering
  \centerline{\includegraphics[width=8 cm]{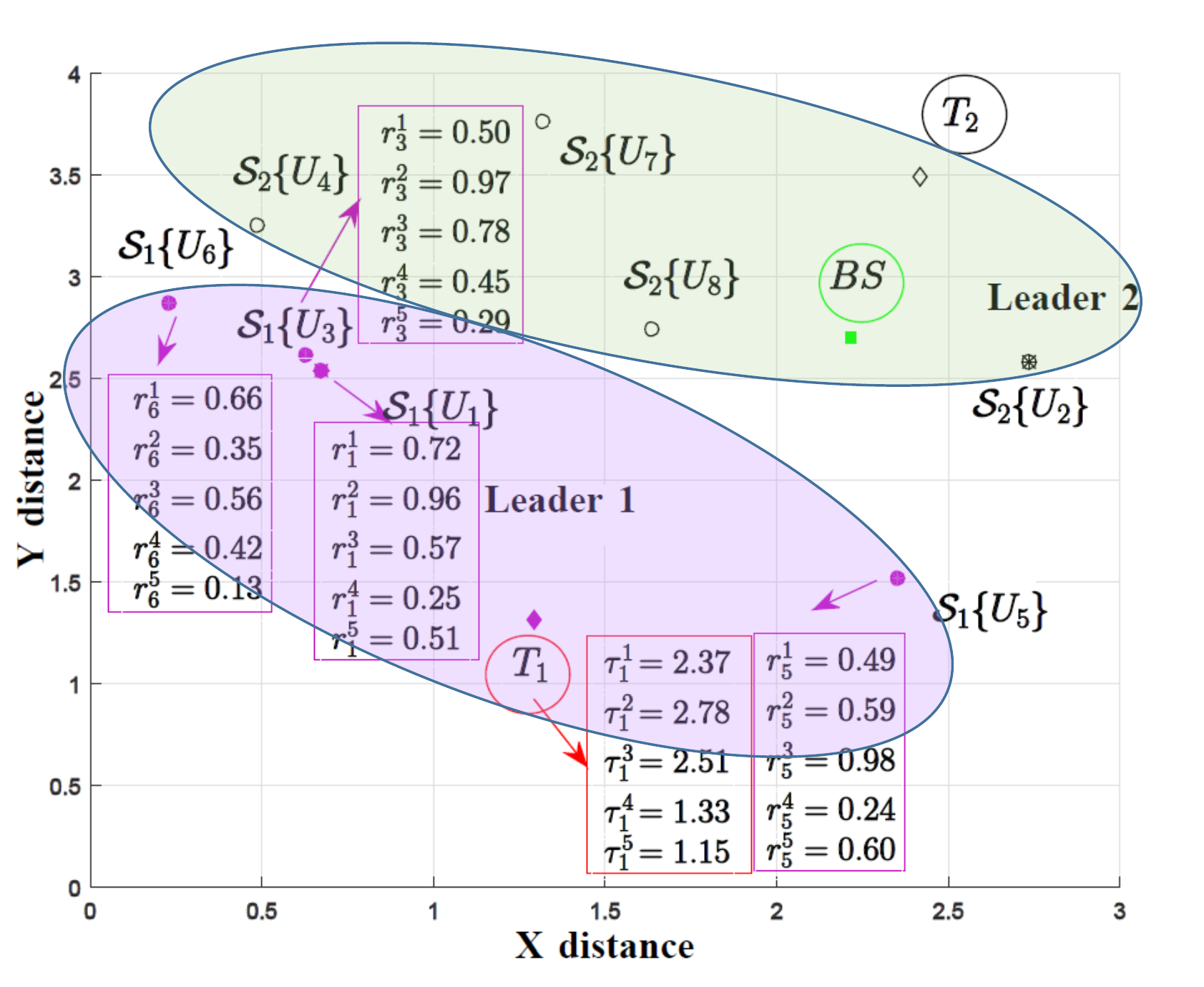}}
    \caption{The stable formed coalitions to complete two identified tasks in an experiment with two leader- and six follower- UAVs}
    \label{fig:Geografic}
\end{figure}
 Figure \ref{fig:Geografic} demonstrates the positions of UAVs in the network and the formed coalitions in a numerical experiment. For this example, the stable formed coalitions are as $\mathcal{S}_1 = \{U_1,U_3,U_5,U_6\}$ and $\mathcal{S}_2 = \{U_2,U_4,U_7,U_8\}$, where coalitions $\mathcal{S}_1$ and $\mathcal{S}_2$ complete the tasks 1 and 2, respectively. In this figure, we show the UAVs by notation of $\mathcal{S}_i\{U_j\}$ that refers to the UAV number $j$ is in coalition $\mathcal{S}_i$. The proposed algorithm has been examined for different scenarios, where the stable coalitions are formed after few rounds. The available resources in coalition $\mathcal{S}_1$ provided by each UAV as well as the required resources to complete $T_1$ are listed in this figure.

Table \ref{Tabel:Coalition1Result} shows the comparison of available resources in coalition $\mathcal{S}_1$ versus the required resources to complete the task encountered by this coalition. As shown in this table, the summation of available resources in coalition $\mathcal{S}_1$ are higher than the required resources for task 1 which guarantees that this task can be completed by the members of coalition $\mathcal{S}_1$. 
Moreover, to evaluate the efficiency of our proposed coalition formation method in a resource constraint network, we define an \emph{Efficiency Factor} as $E.~F.~ = \sum_{j = 1}^{N_r}\{\sum_{i \in \mathcal{S}_1}{r_i^j} / \tau_1^j\}/{N_r}$. This factor evaluates the performance of the formed coalitions in terms of resource allocation efficiency, where the closer value of $E.~F.~$ to 1 means the more efficient the algorithm is in terms of not over-spending the resources for a particular task. In our example, the average $E.~F.~$ value is 1.11, which is fairly close to 1.
\begin{table}[b]
\vspace{-15pt}
  \begin{center}
  \caption{Comparison of the available resources and the required resources in coalition $ \mathcal{S}_1$}
    \begin{tabular}{|c|c|c|c|}
      \hline
      $\begin{array}{c}
        \mbox{Available} \\
        \mbox{resources in~} \mathcal{S}_1 \\
        \sum\limits_{i \in \mathcal{S}_1}{r_i^j}
      \end{array}$
      &
        $\begin{array}{c}
          \mbox{Required} \\
          \mbox{resources} \\
          \mbox{for~} T_1 \\
          ( \tau_1^j)
        \end{array}$
        & $\sum\limits_{i \in \mathcal{S}_1}{r_i^j} \geq \tau_1^j$  &
      {$\sum\limits_{i \in \mathcal{S}_1}\frac{r_i^j}{\tau_1^j}$}\\
      \hline \hline
      \pbox{20cm}{resource 1: 2.37} & 2.37 & $\surd$ & 1.00 \\
      \hline
      resource 2: 2.87 & 2.78 & $\surd$ & 1.03\\
      \hline
       resource 3: 2.90 & 2.51 & $\surd$ & 1.16\\
      \hline
       resource 4: 1.36 & 1.33 & $\surd$ & 1.02\\
      \hline
       resource 5: 1.53 & 1.15 & $\surd$ & 1.33\\
      \hline \hline
      \multicolumn{4}{|c|}{$E.~F.~ = \sum_{j = 1}^{N_r}\{\sum_{i \in \mathcal{S}_1}{r_i^j} / \tau_1^j\}/{N_r} = 1.11$}\\
      \hline
    \end{tabular}
    \label{Tabel:Coalition1Result}
  \end{center}
\end{table}
\begin{figure} [t]
    \psfrag{Efficiency Factor}[][]{\Large{Efficiency Factor}}
    \psfrag{Games Rounds}[][]{\Large{Operation Rounds}}
    \psfrag{Our proposed Coalition formation result}{\Large{Coalition formation results}}
    \psfrag{Choosing close nodes}{\Large{Choosing close UAVs}}
\vspace{-10pt}
     \centerline{\resizebox{!}{6.cm}{\includegraphics{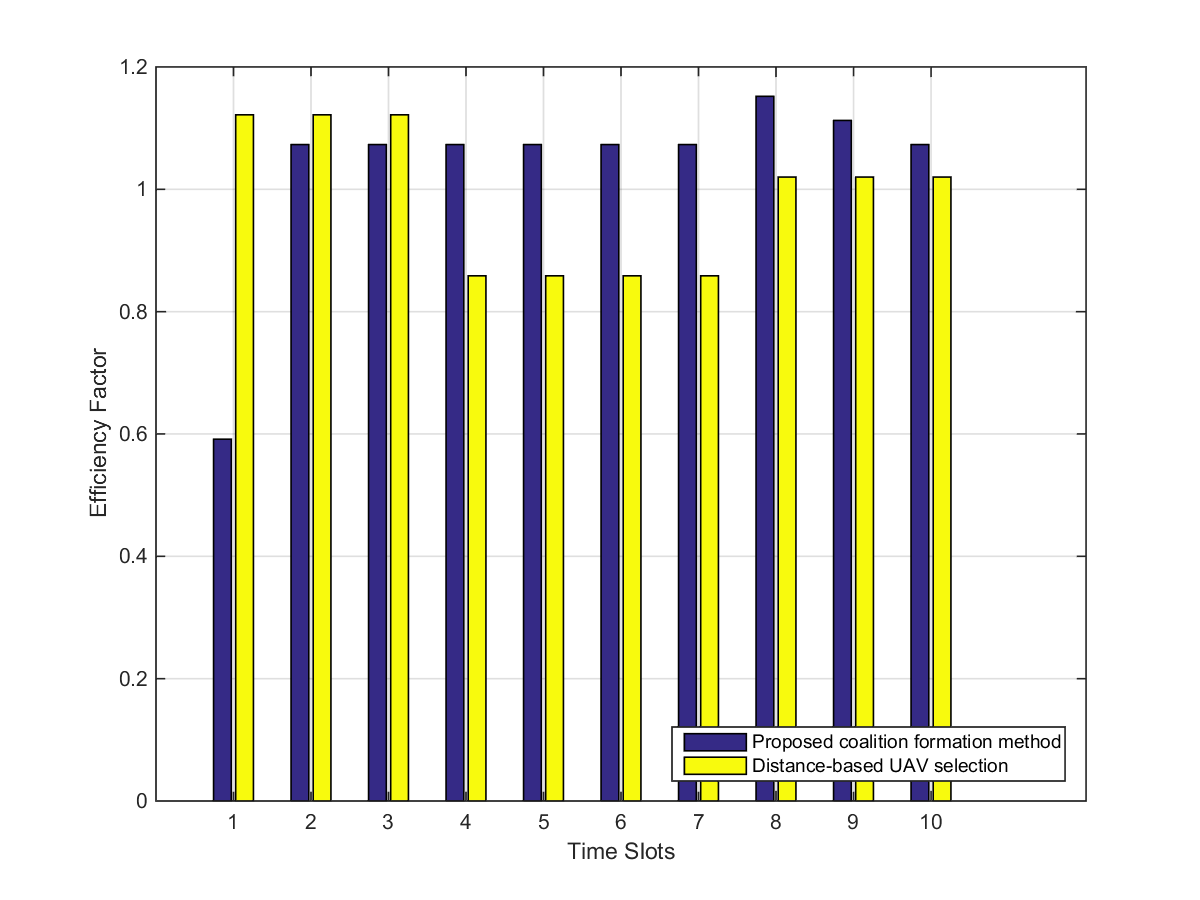}}}
    \caption{Efficiency factor for the proposed coalition formation method compared to the case of selecting the closest UAVs. }
    \label{fig:comparison}
\end{figure}
In Fig. \ref{fig:comparison}, the efficiency factor of the proposed method after forming stable coalitions  is compared to the scenario in which the closest UAVs are assigned to the targets without considering the resources offered by these UAVs. As shown in this figure, our proposed method outperforms the case of distance-based UAVs selection for different system settings over the course of time.

Figure \ref{fig:credit} evaluates the performance of our proposed coalition formation method in identifying the potential selfish UAVs by showing the change in cooperative credit of six follower-UAVs over time. The credits are normalized to be in the range of $[0,1]$. In this scenario, we assume that UAVs, $U_5$ and $U_6$ are selfish in the sense that they do not consume the resources they initially committed after joining a coalition. As seen in Fig. \ref{fig:credit}, the credits of these selfish UAVs, $C_5$ and $C_6$ significantly decrease over time, meaning that these selfish users will not be selected by the leaders in the next rounds of coalition formation. It is worth mentioning that other factors including changes in UAVs' location and dynamic nature of the task requirements can also play a role in small variations in agents' credits that may result in credit reduction for trustable agents. The slight reduction of credits of $U_3$ and $ U_4$ is an example of this fact.
\begin{figure} [t]
\vspace{-5pt}
    \psfrag{Credit }[][]{\Large{Credit}}
    \psfrag{Games Rounds}[][]{\Large{Operation Rounds}}
    \psfrag{Cr 1}{$C_1$}
    \psfrag{Cr 2}{$C_2$}
    \psfrag{Cr 4}{$C_4$}
    \psfrag{Cr 5}{$C_5$}
    \psfrag{Cr 6}{$C_6$}
     \centerline{\resizebox{!}{6.cm}{\includegraphics{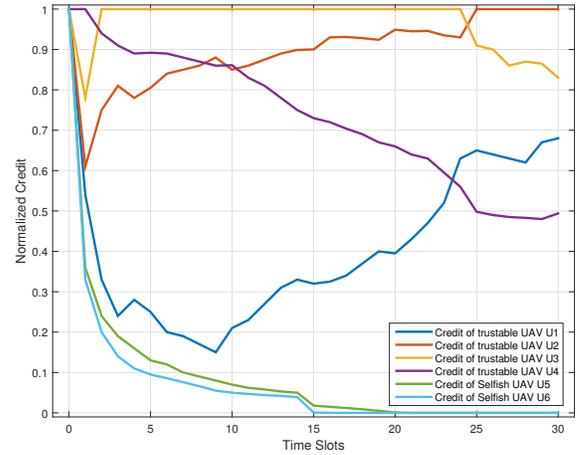}}}
    \caption{The change in cooperative credit of UAVs based on their cooperative/selfish behavior in resource sharing. UAVs 5 and 6 are assumed to be selfish.}
    \label{fig:credit}
\end{figure}

\section{Conclusion}\label{sec:Conclusion}
A leader-follower coalition formation game is developed for distributed task allocation and optimizing the cooperative communication between the detected targets and the base station in a heterogeneous network of UAVs. A reputation-based mechanism is developed to monitor the cooperative behavior of UAVs and filter out the selfish UAVs who have not accumulated sufficient collaboration credits. The proposed methodology enables optimizing several factors including the timely completion of the tasks, and preserving the network resources from the leader's perspectives, while it benefits follower UAVs by lowering their travel times to join the coalitions. 
The simulation results show the convergence of the proposed method in forming stable coalitions with high resource efficiency factors to complete the encountered tasks.

\section{ACKNOWLEDGMENT OF SUPPORT AND DISCLAIMER}
(a) Contractor acknowledges Government's support in the publication of this paper. This material is based upon work funded by AFRL under AFRL Contract No. FA8750-16-3-6003. (b) Any opinions, findings and conclusions or recommendations expressed in this material are those of the author(s) and do not necessarily reflect the views of AFRL.

\appendix \label{Appendix:One}
\section{Appendix A}
\textbf{Appendix A:} The optimization problem becomes feasible if and only if the matrix $\bk \bk^H - t\bH_{UB}^H\bH_{UB}$ is nonnegative definite. Since the matrix $\bQ = \bH_{UB}^H\bH_{UB}$ is positive semi-definite, the matrix $\bQ^{-\frac{1}{2}}(\bk \bk^H - t\bH_{UB}^H\bH_{UB})\bQ^{-\frac{1}{2}} = \bQ^{-\frac{1}{2}}\bk \bk^H \bQ^{-\frac{1}{2}} - t \textbf{I}$ must be nonnegative definite. That results $t \leq \lambda_{max}(\bQ^{-\frac{1}{2}}\bk \bk^H \bQ^{-\frac{1}{2}}) = \bk^H \bQ^{-1}\bk $.
\bibliographystyle{IEEEtran}
\bibliography{IEEEabrv,references}
\end{document}